\documentclass[review]{elsarticle}
\pdfoutput=1
\usepackage[top=3cm,bottom=3cm,left=4cm,right=4cm]{geometry}
\usepackage{algorithm}
\usepackage{algorithmicx}
\usepackage[noend]{algpseudocode}
\usepackage{algpascal}
\usepackage{graphicx}
\usepackage{float}
\usepackage{subfig}
\usepackage{amsmath,amsfonts,amsthm,amscd, amssymb}
\usepackage{latexsym}
\usepackage[all]{xy}
\usepackage{tikz}
\usepackage{placeins}
\usepackage{amsmath}
\usepackage{amssymb}
\usepackage{color}
\definecolor{red}{RGB}{255,0,0}
\definecolor{blue}{RGB}{0,0,255}
\definecolor{green}{RGB}{0,255,0}

\newcommand {\abs}[1]  {\left\vert#1\right\vert}
\newcommand {\set}[1]  {\left\{#1\right\}}

\newcommand {\npc}     {\textsc{NP}\textrm{-complete}}
\newcommand {\nph}     {\textsc{NP}\textrm{-hard}}

\newcommand {\nohyp}   {\textsc{P}=\textsc{NP}}

\newcommand {\comment}[1] {\textcolor{red}{$\bigstar$ #1 $\bigstar$\\}}
\renewcommand{\comment}[1] {}

\newtheorem{theorem}{Theorem}
\newtheorem{lemma}[theorem]{Lemma}

\newcommand{\lb}{{\lfloor n/2 \rfloor}}
\newcommand{\gb}{{\lceil n/2 \rceil}}

\newcommand{\overbar}[1]{\mkern 1.5mu\overline{\mkern-1.5mu#1\mkern-1.5mu}\mkern 1.5mu}

\begin{document}

\begin{frontmatter}

\title{A Simple Extension of Dirac's Theorem on Hamiltonicity\tnoteref{t1}}
\tnotetext[t1]{This work is supported by TUBITAK-CNRS under grant no. 114E731.}
\author[yb]{Yasemin B\"{u}y\"{u}k\c{c}olak\corref{yasemin}}
\ead{y.buyukcolak@gtu.edu.tr}
\address[yb]{Department of Mathematics, Gebze Technical University, Kocaeli, Turkey}
\author[dg]{Didem G\"{o}z\"{u}pek}
\ead{didem.gozupek@gtu.edu.tr}
\address[dg]{Department of Computer Engineering, Gebze Technical University, Kocaeli, Turkey}
\author[yb]{Sibel \"{O}zkan}
\ead{s.ozkan@gtu.edu.tr}
\author[telhai,bu]{Mordechai Shalom\fnref{mordo}}
\fntext[mordo]{The work of this author is supported in part by the TUBITAK 2221 Programme.}
\cortext[yasemin]{Corresponding author.}
\ead{cmshalom@telhai.ac.il}
\address[telhai]{TelHai College, Upper Galilee, 12210, Israel}
\address[bu]{Department of Industrial Engineering, Bo\u{g}azi\c{c}i University, Istanbul, Turkey}

\thispagestyle{plain}
\pagestyle{plain}

\begin{abstract}
The classical Dirac theorem asserts that every graph $G$ on $n$ vertices with minimum degree $\delta(G) \ge \gb$  is Hamiltonian. The lower bound of $\gb$ on the minimum degree of a graph is tight. In this paper, we extend the classical Dirac theorem to the case where $\delta(G) \ge \lb$  by identifying the only non-Hamiltonian graph families in this case. We first present a short and simple proof. We then provide an alternative proof that is constructive and self-contained. Consequently, we provide a polynomial-time algorithm that constructs a Hamiltonian cycle, if exists, of a graph $G$ with $\delta(G) \ge \lb$, or determines that the graph is non-Hamiltonian. Finally, we present a self-contained proof for our algorithm which provides insight into the structure of Hamiltonian cycles when $\delta(G) \ge \lb$ and is promising for extending the results of this paper to the cases with smaller degree bounds.
\end{abstract}

\begin{keyword}
Hamiltonian cycle, sufficiency condition, minimum degree, Dirac theorem, extension, self-contained, constructive, graph algorithms.
\end{keyword}

\end{frontmatter}

\section{Introduction}\label{sec:intro}
A cycle passing through every vertex of a graph $G$ exactly once is called a \emph{Hamiltonian cycle} of $G$, and a graph containing a Hamiltonian cycle is called \emph{Hamiltonian}. Finding a Hamiltonian cycle in a graph is a fundamental problem in graph theory and has been widely studied. In 1972, Karp \cite{Karp72} proved that the problem of determining whether a given graph is Hamiltonian is $\npc$. Hence, finding sufficient conditions for Hamiltonicity has been an interesting problem in graph theory.

An important sufficient condition for Hamiltonicity proved in 1952 by Dirac \cite{Dirac52} is that
every graph on $n$ vertices with minimum degree at least $\lceil n/2 \rceil$ is Hamiltonian.
This lower bound on the minimum degree is tight, i.e., for every $k < \lceil n/2 \rceil$, there is a non-Hamiltonian graph with minimum degree $k$. In 1960, Ore \cite{Ore60} proved that if for all distinct nonadjacent pairs of vertices  $u$ and $v$ of a graph $G$, the sum of degrees of $u$ and $v$ is at least the order of $G$, then $G$ is Hamiltonian. In 1976, Bondy and Chv\'{a}tal \cite{BondyChvatal76} proved that a graph $G$ is Hamiltonian if and only if its closure is Hamiltonian.

Some additional sufficient conditions have been found  for special graph classes. In 1966, Nash-Williams \cite{NashWilliams66}  proved that every $k$-regular graph on $2k+1$ vertices is Hamiltonian. In 1971, Nash-Williams \cite{NashWilliams71} also  proved that a $2$-connected graph of order $n$ with independence number $\beta$ and minimum degree at least $\max\{(n+2)/3,\beta\}$ is Hamiltonian. Note that this result is stronger than the classical Dirac theorem. However, since finding the independence number of a graph is in general $\nph$, this result does not yield an efficient algorithm; i.e., the sufficiency condition cannot be tested in polynomial-time unless $\nohyp$.

The Rahman-Kaykobad condition given in \cite{RahmanKaykobad05} is a relatively new condition that helps to determine the Hamiltonicity of a given graph $G$: The condition is that for every two nonadjacent pair of vertices $u, v$ of $G$, we have $d(u) + d(v) + dist(u,v) > |V|$, where $d(v)$ denotes the degree of $v$ and $dist(u,v)$ denotes the length of a shortest path between $u$ and $v$. In 2005, Rahman and Kaykobad \cite{RahmanKaykobad05} proved that a connected graph satisfying the Rahman-Kaykobad condition has a Hamiltonian path. In 2007, Mehedy et al. \cite{Mehedy07} proved that for a graph $G$ without cut edges and cut vertices and satisfying the Rahman-Kaykobad condition, the existence of a Hamiltonian path with endpoints $u$ and $v$ and $dist(u,v) \geq 3$ implies that $G$ is Hamiltonian. In \cite{Li06} and \cite{LiLi07}, it is proven  that if $G$ is a $2$-connected graph of order $n \geq 3$ and $d(u) + d(v) \geq n - 1$ for every pair of vertices $u$ and $v$ with $dist(u,v) = 2$, then $G$ is Hamiltonian or a member of a given non-Hamiltonian graph class.

Another important property of graphs related with Hamiltonicity is toughness. It is easy to see that being $1$-tough is a necessary condition for Hamiltonicity. In 1978, Jung \cite{Jung78} proved that a $1$-tough graph $G$ on $n > 11$ vertices with  the sum of degrees of non-adjacent vertices $u$ and $v$ at least $n-4$ is Hamiltonian. In 1990,  Bauer, Morgana and Schmeichel \cite{Bauer89} provided a simple proof of Jung's theorem for graphs with more than 15 vertices.  However, in 1990 Bauer, Hakimi and Schmeichel \cite{Bauer90} proved that recognizing $1$-tough graphs is $\nph$. On the other hand, in 2002 Bauer et al. \cite{Bauer02} presented a constructive proof of Jung's theorem for graphs on more than 15 vertices. Furthermore, in 1992 H\"{a}ggkvist \cite{Hagg92} independently showed that for graphs on $n$ vertices with $\delta(G) \ge n/2 - k$, the existence of a Hamiltonian cycle can be recognized in time O($n^{5k}$) where $k \geq 0$ is any fixed integer.

In this paper, we first prove that a graph $G$ with $\delta(G)\ge \lb$ is Hamiltonian except two specific families of graphs. We first provide a simple proof using Nash-Williams theorem \cite{NashWilliams71}. We then provide an alternative proof, which is simple, constructive, and self-contained. Using the constructive nature of our proof, we propose a polynomial-time algorithm that, given a graph $G$ with $\delta(G) \ge \lb$, constructs a Hamiltonian cycle of $G$, or says that $G$ is non-Hamiltonian. The main distinction of our work from \cite{Mehedy07} is that we propose a sufficient condition for Hamiltonicity by using condition $\delta(G) \ge \lb$ and provide explicit non-Hamiltonian graph families, whereas \cite{Mehedy07} uses the Rahman-Kaykobad condition. Our proof also provides a novel insight into the pattern of vertices in a Hamiltonian cycle. We believe that this insight will play a pivotal role in extending our current results to a more general case. Notice that \cite{Li06} and \cite{LiLi07} shows the same non-Hamiltonian graph classes as in our work. However, unlike \cite{Li06} and \cite{LiLi07}, we obtain these graph classes constructively as a result of the nature of our proof. On the other hand, our main distinction from \cite{Hagg92} is that, \cite{Hagg92} shows the polynomial-time recognizability of only the existence of a Hamiltonian cycle under such a minimum degree condition, whereas we construct a Hamiltonian cycle (if exists) in addition to determining whether a Hamiltonian cycle exists when $\delta(G) \ge \lfloor n/2 \rfloor$. In other words, \cite{Hagg92} leads only to a decision algorithm, whereas we provide a construction algorithm.

Jung's theorem states that if $G$ is a $1$-tough graph on $n\ge 11$ vertices such that $d(x)+d(y) \ge n-4$ for all distinct nonadjacent vertices $x,y \in V(G)$, then $G$ is Hamiltonian. Bauer provided a constructive proof of Jung's theorem in \cite{Bauer89}. If $\delta(G) \ge \lb$, a constructive algorithm can then be designed as follows: Run the decision algorithm of H\"{a}ggkvist \cite{Hagg92} to determine whether there is a Hamiltonian cycle. If yes, then construct a Hamiltonian cycle using the constructive proof of Bauer in \cite{Bauer89}. However, this approach has three main drawbacks: Unlike this work, (i) such an approach fails to specify non-Hamiltonian graph families under the minimum degree condition $\delta(G) \ge \lb$, (ii) it is not self-contained, (iii) it does not explicitly provide a polynomial-time algorithm. Furthermore, this paper provides a shorter and simpler proof than \cite{Bauer89}. Finally, our algorithm can be used to generate all Hamiltonian cycles under the condition  $\delta(G) \ge \lb$.

\section{Preliminaries}\label{sec:prelim}
We adopt \cite{West2000}  for terminology and notation not defined here. A graph  $G = (V,E)$ is given by a pair of a vertex set $V = V(G)$ and an edge set $E = E(G)$ where $uv \in E(G)$ denotes an edge between two vertices $u$ and $v$. In this work, we consider only simple graphs, i.e., graphs without loops or multiple edges. In particular, we use $G_n$ to denote a not necessarily connected simple graph on $n$ vertices. $\abs{V(G)}$ denotes the order of $G$ and $N(v)$ denotes the neighborhood of a vertex $v$ of $G$. In addition, $\delta(G)$  denotes the minimum degree of $G$ and the \emph{distance} $dist(u,v)$ between two vertices $u$ and $v$ is the length of a shortest path joining $u$ and $v$, whereas the \emph{diameter} of $G$, denoted by $diam(G)$, is the maximum distance among all pairs of vertices of $G$. If $P=x_{0}x_{1}x_{2} \dots x_{k}$ is a path, then we say that $x_{i}$ \emph{precedes} (resp. \emph{follows}) $x_{i+1}$ (resp. $x_{i-1}$).

Given two graphs $G=(V,E)$ and $G'=(V',E')$, we define the following binary operations.
The \emph{union} $G\cup G'$ of $G$ and $G'$ is the graph obtained by the union of their vertex and edge sets, i.e., $G \cup G' = (V \cup V', E \cup E')$. When $V$ and $V'$ are disjoint, their union is referred to as the \emph{disjoint union} and denoted by $G + G'$. The \emph{join} $G \vee G'$ of $G$ and $G'$ is the disjoint union of graphs $G$ and $G'$ together with all the edges joining $V$ and $V'$. Formally,  $G \vee G' = (V \cup V', E \cup E' \cup \{V \times V'\})$. $K_n$ and $\overbar{K_n}$ denote the complete and empty graph, respectively, on $n$ vertices.

The classical Dirac theorem is as follows:
\begin{theorem}\label{di}\textbf{(Dirac \cite{Dirac52})}
If $G$ is a graph of order $n \geq 3$ such that $\delta(G) \geq  n/2$, then
$G$ is Hamiltonian.
\end{theorem}

We now present the main theorem of this paper:
\begin{theorem}\label{thm:ExtendedDirac}
Let $G$ be a connected graph of order $n \geq 3$ such that $\delta(G) \geq \lb$. Then $G$ is Hamiltonian
unless $G$ is the graph $K_{\lceil n/2 \rceil} \cup K_{\lceil n/2 \rceil}$ with one common vertex or a graph  $\overbar{{K}}_{\lceil n/2 \rceil} \vee G_\lb$  for odd $n$.
\end{theorem}

The constructive nature of our proof for Theorem \ref{thm:ExtendedDirac} given in Section \ref{sec:proofMain} yields the following result:

\begin{theorem}\label{thm:Algorithm}
Let $G$ be a graph of order $n \geq 3$ such that $\delta(G) \geq \lb$. Then
there is a polynomial-time algorithm that determines whether a Hamiltonian cycle exists and finds a Hamiltonian cycle in $G$, if such a cycle exists.
\end{theorem}

We provide the proof of Theorem \ref{thm:Algorithm} in Section \ref{sec:proofAlg}.

\section{Proofs of Theorem \ref{thm:ExtendedDirac}}\label{sec:proofMain}
In this section, we prove Theorem \ref{thm:ExtendedDirac} of this paper, which extends  the classical Dirac theorem. Using Nash-William's theorem \cite{NashWilliams71}, we first provide a simple proof.
\begin{lemma}
\cite{NashWilliams71} Let $G$ be a 2-connected graph of order $n$ with independence number $\beta(G)$ and minimum degree $\delta(G)$. If $\delta(G) \ge \max((n+2)/3, \beta(G))$, then $G$ is Hamiltonian. \label{lemma:nashwilliams}
\end{lemma}

\begin{proof} [Proof-1 of Theorem \ref{thm:ExtendedDirac}]
Let $G$ be a non-Hamiltonian graph on $n$ vertices with $\delta(G) \geq \lb$. Since for even $n$, non-Hamiltonicity would contradict Dirac's theorem, $n$ must be odd. Let $n=2r+1$ where $r \in \mathbb{Z}^+$.

First, consider the case that $G$ is not 2-connected, and consider a cut vertex $v$. Let $G[V(G)\setminus \{v\}]$ have $k$ connected components $G_i$ where $1 \leq i \leq k$. Since $\delta(G) \ge \lb = r$, we have $\abs{V(G_i)} - 1 \geq \delta(G_i) \geq r-1$, thus $\abs{V(G_i)} \geq r$ for $1 \leq i \leq k$. Since $n=2r+1$, we have $k = 2$ and $\abs{V(G_1)}=\abs{V(G_2)}=r$. Therefore, the only 1-connected graph $G$ with $\delta(G) \ge \lb$ is the graph $K_{\lceil n/2 \rceil} \cup K_{\lceil n/2 \rceil}$ with one common vertex.

Now consider the case that $G$ is 2-connected. If $n \geq 7$, $\delta(G) \geq (n-1)/2 \geq (n+2)/3$. If $\delta(G) \ge \beta(G)$, then $\delta(G) \ge \max((n+2)/3, \beta(G))$ and $G$ is Hamiltonian due to Lemma \ref{lemma:nashwilliams}, contradiction. Therefore, $\beta (G) > \delta(G) \ge (n-1)/2$ and hence $\beta(G) \ge (n+1)/2$. In other words, there is an independent set $S$ with size $(n+1)/2$. Since $\delta(G)\ge (n-1)/2$, each vertex in the independent set $S$ has to be adjacent to each vertex in $V(G)\setminus S$, which results in the graph  $\overbar{{K}}_{\lceil n/2 \rceil} \vee G_\lb$.

For $n=3$, the only 2-connected graph is a cycle on $3$ vertices, which is Hamiltonian. For $n=5$ consider the minimal graphs $G$ with $\delta(G) \geq 2$, where by minimality we imply that removal of an edge violates the degree condition $\delta(G) \geq 2$. The vertices $U$ of degree more than $2$ in $G$ constitute an independent set because otherwise it would not be minimal. Let $\overbar{U}=V(G) \setminus U$. Since $n=5$ and $\delta(G)\ge 2$, We have $\abs{\overbar{U}} \geq 3$, i.e. $\abs{U} \leq 2$. If $U=\emptyset$, then $G$ is a cycle, a contradiction. If $U=\set{u}$ and $d(u)=4$, the degree sequence of $G[\overbar{U}]$ is $(1,1,1,1)$ and $u$ is a cut vertex, contradicting 2-connectivity. If $U=\set{u}$ and $d(u)=3$ the degree sequence of $G[\overbar{U}]$ is $(2,1,1,1)$, a contradiction since such a graph with odd degree sum does not exist. If $\abs{U}=2$, $G$ contains a $K_{2,3}$ with possible additional edges between vertices of $\overbar{U}$. If there are no such additional edges, $G$ is a $K_{2,3}$ and satisfies the claim. Otherwise, if there is such an additional edge, one can build a Hamiltonian path of $K_{2,3}$ and close it to a cycle using this edge, contradiction.
\end{proof}

We now present not only simple, but also self-contained and constructive proof. Our proof is inspired by the proof of the following theorem in \cite{NashWilliams66}.

\begin{lemma}
\cite{NashWilliams66} Every $k$-regular graph on $2k +1$ vertices is Hamiltonian.
 \label{lemma:nashwilliams2}
\end{lemma}

We now give the proof for Theorem \ref{thm:ExtendedDirac} as follows:
\begin{proof}[Proof-2 of Theorem \ref{thm:ExtendedDirac}]
For $n = 2r$  where $r \in \mathbb{Z}^+$, the result holds by Theorem \ref{di}. Hence, we assume that $n = 2r + 1$ and $\delta(G) \geq r$.
First, we consider the graph $G'$ obtained by adding a new vertex $y$ to $G$ and connecting it to all other vertices. The graph  $G'$ has $|V(G')| = 2r+2$ vertices and minimum degree $\delta(G') \geq r+1$. By Theorem \ref{di}, $G'$ has a Hamiltonian cycle $C$. The path $P$ obtained by the removal of $y$ from $C$ is a Hamiltonian path of $G$. Let $P=x_{0}x_{1} \dots x_{2r}$.

Suppose $G$ has no Hamiltonian cycle. That is, $x_0$ and $x_{2r}$ are not adjacent.
Then, we observe the following facts:
\begin{enumerate}
  \item \textit{If $x_0$ is adjacent to $x_i$, then $x_{2r}$ is not adjacent to $x_{i-1}$.}
Otherwise, the closed trail  $x_0x_1 \dots x_{i-1}x_{2r}x_{2r-1}x_{2r-2} \dots x_{i}x_0$ is a Hamiltonian cycle. \label{fact1}
  \item \textit{If $x_0$ is not adjacent to $x_i$, then $x_{2r}$ is adjacent to $x_{i-1}$.}
By the first fact, $N(x_{2r}) \subseteq X = \set{x_{i-1}|x_i \notin N(x_0), i \geq 1}$. Since $\abs{X}=r$ and $d(x_{2r}) = r$, then $N(x_{2r})=X$.
  \item \textit{Every pair of non-adjacent vertices  $x_i$ and $x_j$, where $0 \leq i,j \leq 2r$, has at least one common neighbor.} This is because $N(x_i) \subseteq V(G) \setminus \set{x_i, x_j}$, $N(x_j) \subseteq V(G) \setminus \set{x_i, x_j}$, $d(x_i) \geq r$, and $d(x_j) \geq r$. Note that this implies $diam(G) = 2$.
\end{enumerate}

Assume that $d(x_0) > r$. By Fact \ref{fact1}, the vertex  $x_{2r}$ cannot be adjacent to at least $r+2$ vertices of $G$ including itself. Since $n = 2r + 1$, it implies that $d(x_{2r}) \leq r-1$, contradicting with $\delta(G) \geq r$. Similarly, if $d(x_{2r}) > r$, then $d(x_{0}) \leq r-1$ by the symmetry of the Hamiltonian path $P$, again contradiction. Therefore, $d(x_0) = d(x_{2r}) = r$.

We now  consider two disjoint and complementary cases:
\begin{enumerate}
  \item {\textbf{\underline{$N(x_0) \cup N(x_{2r}) = V(G) \setminus \set{x_0, x_{2r}}$}}:} \label{case1}
  By this assumption and  Fact 3, $x_0$ and $x_{2r}$ have exactly one common neighbor $x_k$. Then $x_{k-1}$ is not adjacent to $x_{2r}$ but adjacent to $x_0$. Proceeding in the same way, we conclude that $N(x_0) = \set{x_1, \ldots, x_k}$ and $N(x_{2r}) = \set{x_k,\ldots,x_{2r-1}}$. Since $d(x_0)=d(x_{2r})=r$, we conclude that $k = r$. Let $i \in [r+1,2r-1]$ and $i_0 \in [0, r-1]$. If $x_{i_0}x_i \in E$, the cycle $x_{i_0}x_{i_0-1} \dots x_0x_{i_0+1}x_{i_0+2} \dots x_{i-1}x_{2r}x_{2r-1} \dots x_{i}x_{i_0}$ is a Hamiltonian cycle of $G$. Therefore, for every $i \in [r+1,2r-1]$ and every $i_0 \in [0,r-1]$, $x_i$ and $x_{i_0}$ are non-adjacent. Then $G = K_{\lceil n/2 \rceil} \cup K_{\lceil n/2 \rceil}$ with one common vertex $x_r$. Note that $G$ is not Hamiltonian since it contains a cut vertex, namely $x_r$.

  \item {\textbf{\underline{$N(x_0) \cup N(x_{2r}) \neq V(G) \setminus \set{x_0, x_{2r}}$}}:}\label{case2}
  Then there is an $i_0 \in [2, 2r-2]$ such that $x_{i_0 + 1}$ is adjacent to $x_{0}$, but $x_{i_0}$ is not. By Fact 2, $x_{i_0 - 1}$ is adjacent to $x_{2r}$.
  Hence, we have a $(2r)$-cycle $x_{i_0-1}x_{i_0-2} \dots x_0x_{i_0+1}x_{i_0+2} \dots x_{2r}x_{i_0-1}$ which does not contain $x_{i_0}$, say $C$. We rename the vertices of $C$ such that $y_1y_2 \dots y_{2r}$ and $y_0=x_{i_0}$. If $y_0$ is adjacent to two consecutive vertices of $C$, then $G$ is Hamiltonian. Therefore, $y_0$ is not adjacent to two consecutive vertices of $C$. Combining this with the fact that $d_{y_0} \geq r$ we conclude that $d(y_0)=r$ and $y_0$ is adjacent to every second vertex of $C$. Without loss of generality, let $N(y_0)=\set{y_1, y_3, ... , y_{2r-1}}$.  Observe that by replacing $y_{2i}$ by $y_0$ for some $i \in [1,r]$ we obtain another cycle with $2r$ vertices. Then, by the same argument $N(y_{2i})=\set{y_1, y_3, ... , y_{2r-1}}$ for every $i \in [0,2]$.
  Hence, $G=\overbar{K}_{\lceil n/2 \rceil} \vee G_\lb$ where the vertices with even index form the empty graph $\overbar{{K}}_{\lceil n/2 \rceil}$ and the vertices with odd index form a not necessarily connected graph $G_\lb$.   Notice that $G$ is not Hamiltonian since it contains an independent set with more than half of the vertices, namely $\set{y_0,\ldots,y_{2r}}$.
  \end{enumerate}
\end{proof}

Note that we have obtained the non-Hamiltonian graph classes constructively.

In the following section, by using the proof of Theorem \ref{thm:ExtendedDirac}, we propose a polynomial-time algorithm for finding a Hamiltonian cycle of a given graph $G$ with order $n$ and $\delta(G) \ge \lb$, if it exists, or returns ``none" if it does not exist.

\section{Proof of Theorem \ref{thm:Algorithm}}\label{sec:proofAlg}
\newcommand{\alg}{\textsc{FindHamiltonian}}
\newcommand{\makeCycle}{\textsc{MakeCycle}}
\newcommand{\makeCycleA}{\textsc{MakeTypeACycle}}
\newcommand{\makeCycleB}{\textsc{MakeTypeBCycle}}
\newcommand{\makeCycleC}{\textsc{MakeTypeCCycle}}
\alglanguage{pseudocode}

In this section we present Algorithm $\alg$ that, given a graph $G$, returns either a Hamiltonian cycle $C$ or ``NONE". Although $\alg$ may in general return ``NONE" for a Hamiltonian graph $G$, we will show that this will not happen if $\delta(G) \geq \lb$. $\alg$, whose pseudo code is given in Algorithm \ref{alg:FindHamiltonian}, first tests $G$ for the two exceptional graph families mentioned in Theorem \ref{thm:ExtendedDirac}. Once $G$ passes the tests, the algorithm first builds a maximal path by starting with an edge and then extending it in both directions as long as this is possible. After this stage, the algorithm tries to find a larger path by closing the path to a cycle and then adding to it a new vertex and opening it back to a path. Once the path is closed to a cycle, it is clearly possible to extend it to a larger path, since $G$ is connected. It remains to show that under the conditions of Theorem \ref{thm:ExtendedDirac}, i.e. $\delta(G) \geq \lb$, the algorithm will always be able to construct a cycle from the vertices of $P$. This is done in function $\makeCycle$, which in turn tries three different constructions using the functions $\makeCycleA$ (see Figure \ref{fig:TypeACycle}), $\makeCycleB$ (see Figure \ref{fig:TypeBCycle}), and $\makeCycleC$ (see Figure \ref{fig:TypeCCycle}).

Note that Algorithm \ref{alg:FindHamiltonian} is polynomial since i) Lines 1-5 can be computed in polynomial time, ii) constructing a maximal path in Lines 7-8, constructing a cycle in Line 9, and obtaining a larger maximal path in Lines 10-13 can be done in polynomial time, iii) the loop in Lines 6-14 iterates at most $n$ times. Therefore, it is sufficient to prove the following lemma.

\begin{lemma}
Let $\abs{V(G)}=n$, $\delta(G) \geq \lb$ and $P$ be a maximal path of $G$. If function $\makeCycle$ returns NONE, then either $G$ has a cut vertex or an independent set $I$ with more than $n/2$ vertices constituting a connected component of $\bar{G}$.
\end{lemma}
\begin{figure}[H]
  \centering
  \hspace{2cm}
  \includegraphics[scale=0.4]{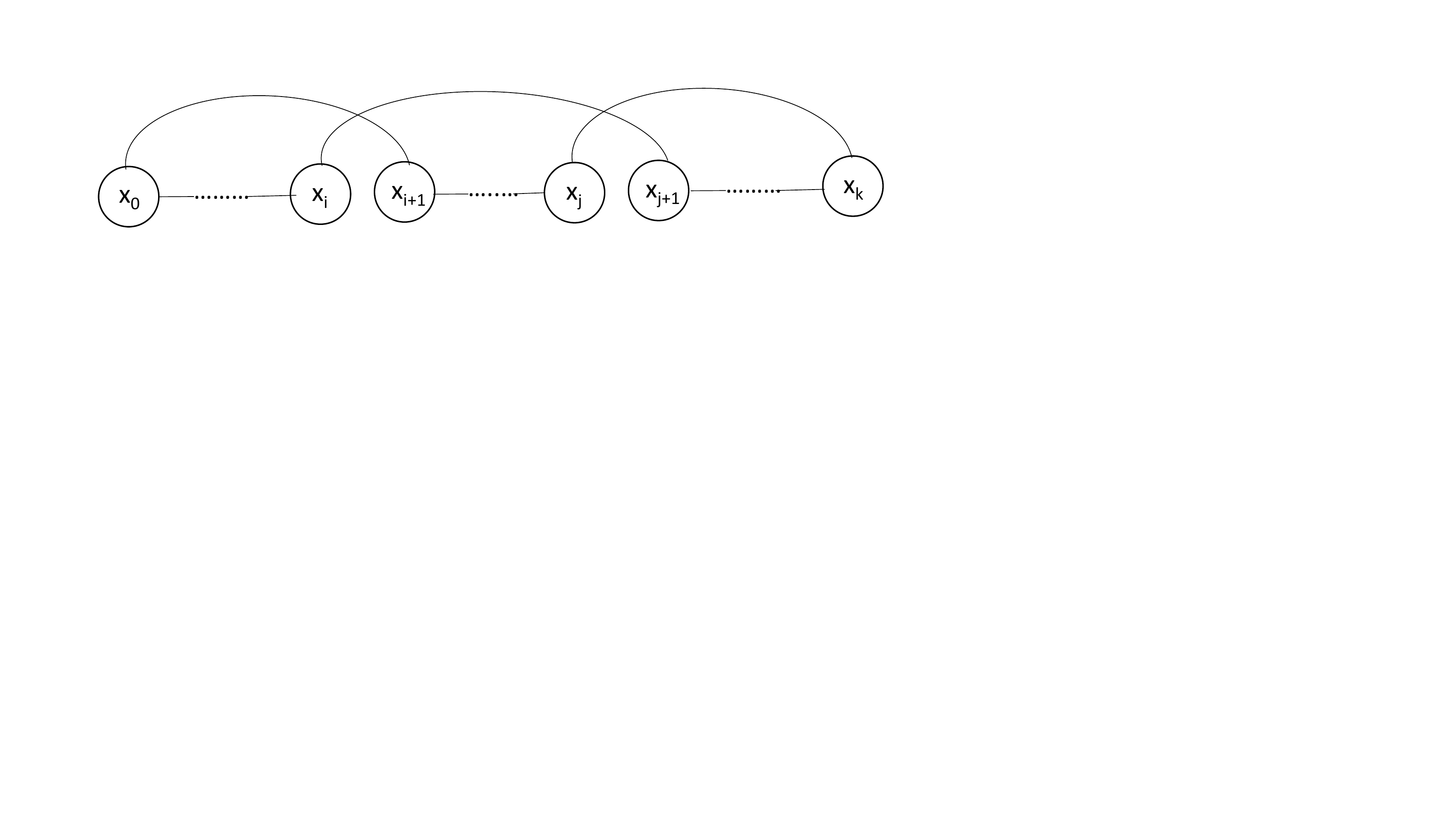}}
  \vspace{-5.5cm}
  \caption{The cycles detected by $\makeCycleA$}{\label{fig:TypeACycle}
\end{figure}
\begin{figure}[H]
  \centering
  \includegraphics[scale=0.4]{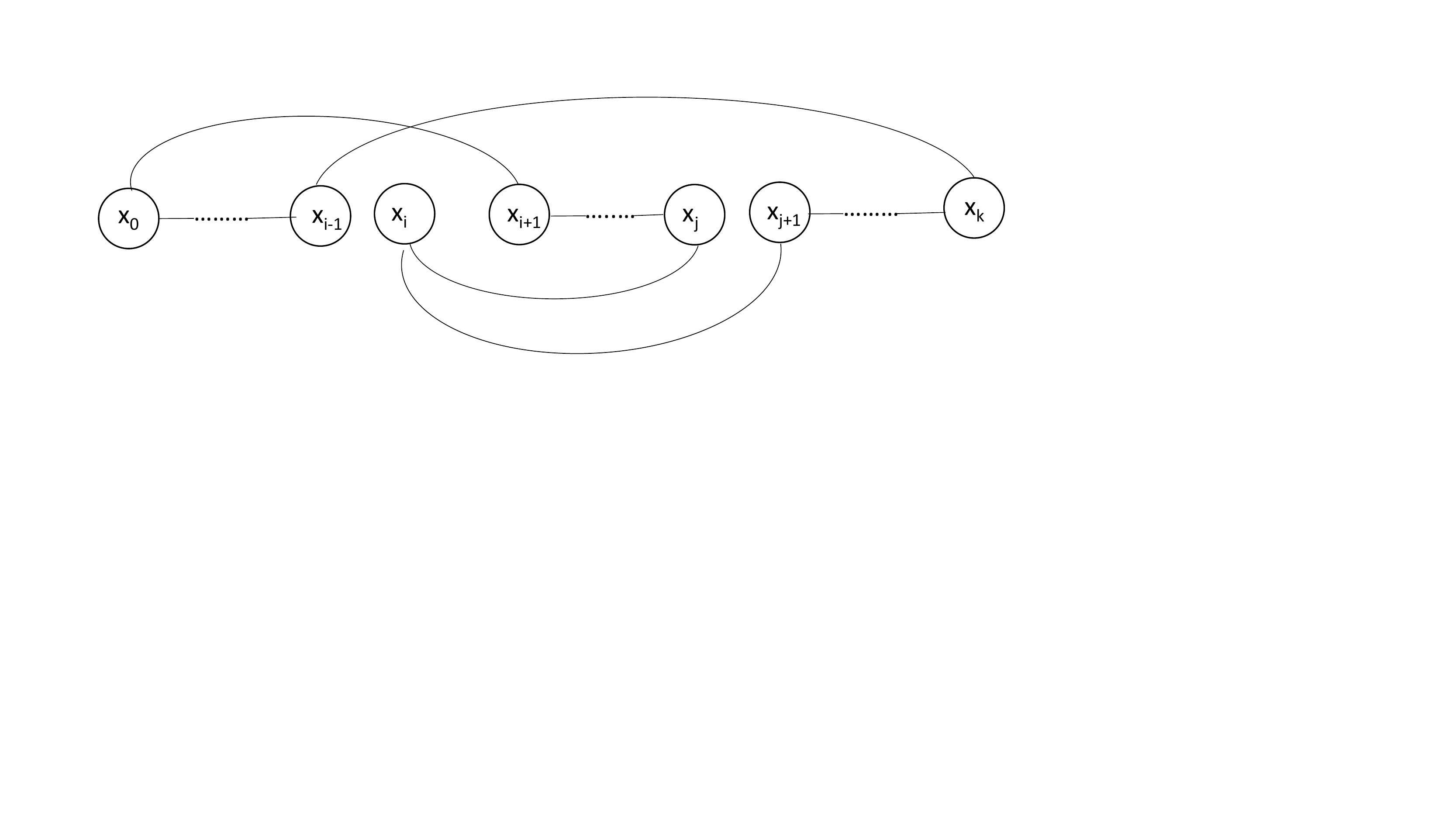}}
  \vspace{-4cm}
  \caption{The cycles detected by $\makeCycleB$}{\label{fig:TypeBCycle}
\end{figure}
\begin{figure}[H]
  \centering
  \includegraphics[scale=0.4]{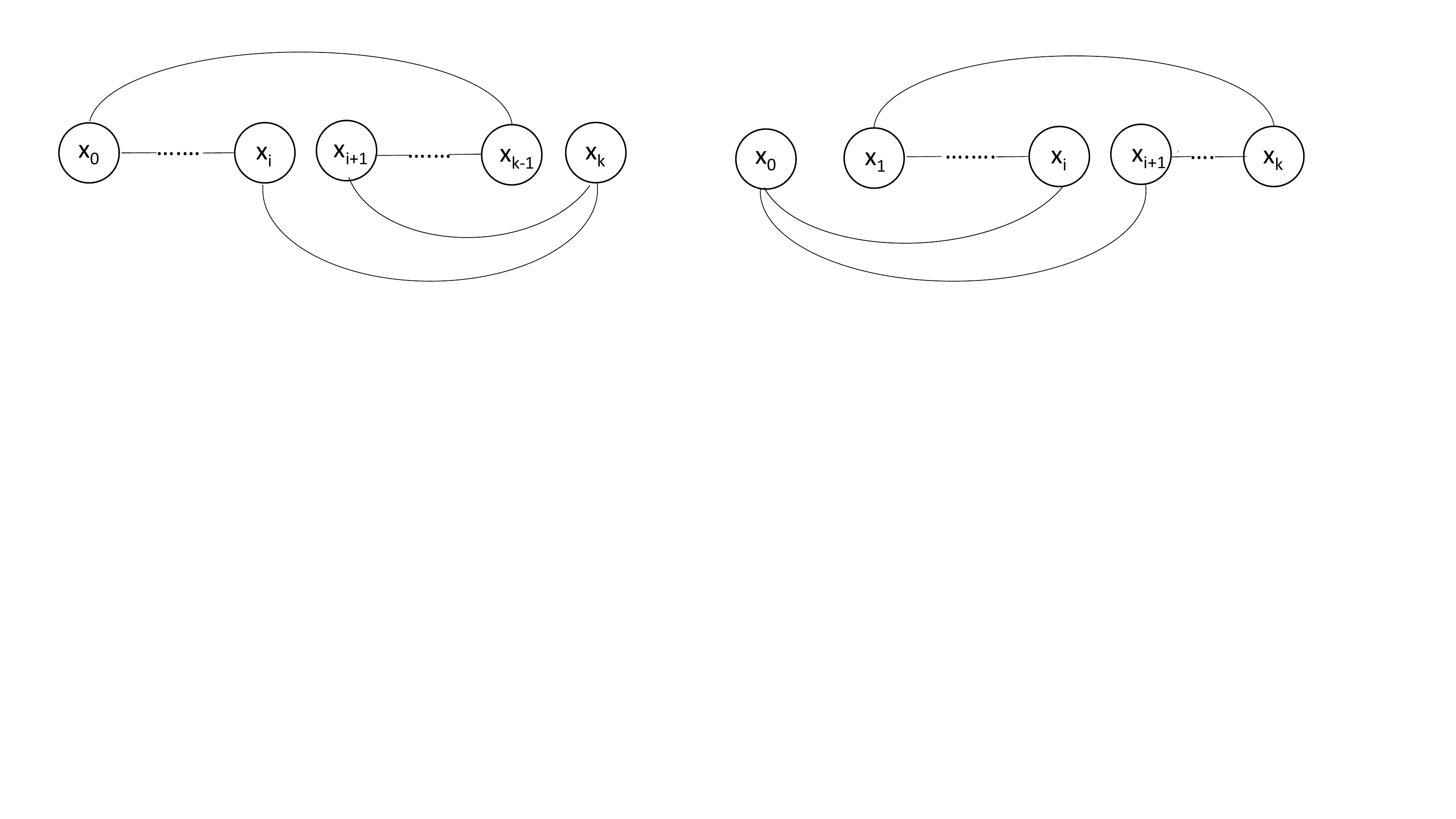}}
  \vspace{-5cm}
  \caption{The cycles detected by $\makeCycleC$}{\label{fig:TypeCCycle}
\end{figure}
\begin{proof}
Let $P=x_0x_1\ldots x_k$ be a maximal path of $G$ for some $k \leq n-1$. Assume that the functions $\makeCycleA$, $\makeCycleB$ and $\makeCycleC$ all return ``NONE". Since $P$ is maximal, $N(x_0),N(x_k) \subseteq V(P)$. Suppose that $x_0 x_k \in E(G)$. Then, setting $i=0$ and $j=k-1$ in function $\makeCycleA$ would detect a cycle. Therefore, $x_0 x_k \notin E(G)$, i.e. $N(x_0),N(x_k) \subseteq V(A)=\set{x_1,\ldots,x_{k-1}}$ where $A$ is the path obtained by deleting the endpoints $x_0$ and $x_k$ of $P$. We partition $V(A)$ by the adjacency of their vertices to $x_0$ and $x_k$. We denote the set of vertices $N(x_0) \setminus N(x_k)$ by $A_{0}$,  $N(x_k) \setminus N(x_0)$ by $A_{k}$, $N(x_0) \cap N(x_k)$ by $A_{0k}$, and the set of vertices $A \setminus (N(x_0) \cup N(x_k))$ by $A_{\overbar{0k}}$. In the sequel we use $a_p$ to denote an arbitrary element of $A_p$ for $p \in \set{0,k, 0k,\overline{0k}}$, and we use regular expression notation for sequences of elements of these sets. In particular, $*$ denotes zero or more repetitions of a pattern.

Suppose that  $x_i \in N(x_k)$ and $x_{i+1} \in N(x_0)$ for some $x_i \in V(A)$. Then, for this value of $i$ and for $j=k-1$, the function $\makeCycleA$ would detect a cycle. We conclude that such a vertex $x_i$ does not exist in $A$. In other words, two consecutive vertices $(x_i, x_{i+1})$ of $A$ do not follow any of the following forbidden patterns: $(a_k,a_0)$, $(a_k,a_{0k})$, $(a_{0k},a_0)$, $(a_{0k},a_{0k})$.

Consider two vertices $x_i, x_j \in A_{0k}$ ($i<j$), with no vertices from $A_{0k}$ between them in $A$. Furthermore, suppose that there are no vertices from $A_{\overbar{0k}}$ between $x_i$ and $x_j$. By these assumptions and due to the forbidden pairs previously mentioned, we have $x_{i+1}  \ldots, x_{j-1} \in A_k$. However, $(x_{j-1}, x_j)$ is also a forbidden pair, contradiction. Therefore, there is at least one vertex from $A_{\overbar{0k}}$ between any two vertices of $A_{0k}$. We conclude that $\abs{A_{\overbar{0k}}} \geq \abs{A_{0k}}-1$. We have
\begin{eqnarray}
n-2 & \geq & k-1 = \abs{A_{0k}}+\abs{A_{k}}+\abs{A_{0}}+\abs{A_{\overbar{0k}}}=(\abs{A_{0k}}+\abs{A_{k}})+(\abs{A_{0k}}+\abs{A_{0}})+
\abs{A_{\overbar{0k}}}-\abs{A_{0k}}\nonumber\\
& \geq & d(x_k)+d(x_0)-1 \geq 2 \delta(G) - 1\nonumber\\
\frac{n-1}{2} & \geq & \delta(G).\nonumber
\end{eqnarray}
Since $\delta(G) \geq \lb \geq \frac{n-1}{2}$, we have $\delta(G) = \frac{n-1}{2}$, and all the inequalities above hold with equality, implying the following:
\begin{enumerate}[a)]
\item {$d(x_0)=d(x_k)=\delta(G)=\frac{n-1}{2}$}, thus $n$ is odd and $\abs{A_{k}} = \abs{A_{0}}$.
\item {$k=n-1$}, thus $V(P)=V(G)$.
\item {$\abs{A_{\overbar{0k}}} = \abs{A_{0k}} - 1$}. There is exactly one vertex of $A_{\overbar{0k}}$ between two consecutive vertices from $A_{0k}$ and there are no other vertices from $A_{\overbar{0k}}$ in $A$.
\end{enumerate}
The vertices between (and including) two consecutive vertices from $A_{0k}$ follow the pattern $(a_{0k} a_k^* a_{\overbar{0k}} a_0^* a_{0k})$. All the vertices before the first vertex from $A_{0k}$ are from $A_{0}$, and all the vertices after the last vertex from $A_{0k}$ are from $A_{k}$. We conclude that $A$ follows the pattern:
\[
a_0^* (a_{0k} a_k^* a_{\overbar{0k}} a_0^*)^* a_{0k} a_k^*
\]

Then, every vertex of $A_{\overbar{0k}}$ is preceded by a neighbor of $x_k$ and followed by a neighbor of $x_0$; in other words, a vertex $x_i \in A_{\overbar{0k}}$ satisfies the condition in Line \ref{lin:TypeBCondition1} of $\makeCycleB$. Since, because our assumption, $\makeCycleB$ does not close a cycle, the condition in Line \ref{lin:TypeBCondition2} is not satisfied for any value of $j$. We conclude that $x_i$ is not adjacent to two consecutive vertices of $A$. Then, the number of neighbours of $x_i$ among $x_1,\ldots,x_{i-1}$ is at most $\lceil \frac{i-1}{2}\rceil$ and the number of neighbours of $x_i$ among $x_{i+1},\ldots,x_{k-1}$ is at most $\lceil \frac{k-1-i}{2}\rceil$. Therefore,
\[
d(x_i) \leq \left\lceil \frac{i-1}{2} \right\rceil + \left\lceil \frac{k-1-i}{2} \right\rceil \leq \frac{i}{2} + \frac{k-i}{2}=\frac{k}{2}=\frac{n-1}{2}=\delta(G).
\]
Since $d(x_i) \geq \delta(G)$, all the inequalities above hold with equality, implying the following:
\begin{enumerate}
\item{Both $i$ and $k$ are even}
\item{$N(x_i)=A_{odd}$ where $A_{odd}=\set{x_1,x_3,\ldots,x_{k-1}}$.}
\end{enumerate}
Since for every $x_i \in A_{\overbar{0k}}$, $x_i$ is even and $N(x_i)=A_{odd}$, we conclude that $A_{\overbar{0k}}$ is an independent set. Recalling that $x_0 x_k \notin E(G)$ and the definition of $A_{\overbar{0k}}$, we conclude that $I = A_{\overbar{0k}} \cup \set{x_0,x_k}$ is an independent set.

We observe in the previous pattern that the set of vertices preceding the neighbours of $x_0$ (i.e. $A_0 \cup A_{0k}$) is $A_0 \cup A_{\overbar{0k}}$, and the set of vertices following the neighbours of $x_k$ (i.e. $A_k \cup A_{0k}$) is $A_k \cup A_{\overbar{0k}}$. Let $x_i$ be a vertex that precedes a neighbour of $x_0$ and let $x_j$ be a vertex that follows a neighbour of $x_k$ with $i<j$. If $x_i x_j \in E$, $\makeCycleA$ can close a cycle since the condition in Line \ref{lin:TypeACondition} is satisfied. Therefore, a pair of adjacent vertices $(x_i, x_j)$ with $i<j$ in $G$ cannot follow one of the following patterns: $(a_{\overbar{0k}},a_{\overbar{0k}})$,$(a_{\overbar{0k}},a_k)$,$(a_0,a_{\overbar{0k}})$,$(a_0,a_k)$. If $\abs{A_{\overbar{0k}}}=0$, then $\abs{A_{0k}}=1$ and $A$ follows the pattern $a_0^* a_{0k} a_k^*$. Since $(a_0,a_k)$ is a forbidden pattern for adjacent vertices, none of the vertices of $A_0$ is adjacent to a vertex in $A_k$. Therefore, the unique vertex $a_{0k}\in A_{0k}$ is a cut vertex of $G$, contradicting our assumption. We conclude that $\abs{A_{\overbar{0k}}} > 0$.

Let $x_i \in A_{\overbar{0k}}$ and $x_j \in A_{odd}=N(x_i)$. If $j < i$ then $x_j \notin A_0$, since otherwise they follow the pattern $(a_0, a_{\overbar{0k}})$ and they are adjacent. Similarly, if $i < j$ then $x_j \notin A_k$. We conclude that, in $A$ all the vertices between two vertices from $A_{\overbar{0k}}$ are from $A_{0k}$. Moreover, all vertices before the first (after the last) vertex from $A_{\overbar{0k}}$ except one vertex from $A_{0k}$ are from $A_k$ (resp. $A_0$). Then $A$ follows the pattern:
\[
a_{0k} a_k^* (a_{\overbar{0k}} a_{0k})^*a_{\overbar{0k}} a_0^* a_{0k}
\]
We now observe that $x_{k-1} \in A_{0k}$, i.e. $x_0 x_{k-1} \in E(G)$. Let $\eta=\abs{A_{k}}=\abs{A_{0}}$. Suppose that $\eta \neq 0$. Then $x_k x_1, x_k x_2 \in E(G)$ and $\makeCycleC$ will close a cycle. Therefore, $\eta=0$, i.e. $A$ follows the pattern:
\[
(a_{0k} a_{\overbar{0k}})^{*} a_{0k}
\]

We conclude that $I$ has $\abs{I}=\frac{n+1}{2}$ vertices, and every vertex of $I$ is adjacent to every vertex of $A_{odd}= A_{0k} = V(G) \setminus I$. Then $I$ is a connected component of $\overbar{G}$.
\end{proof}

\begin{algorithm}[H]
\small
\caption{\alg}\label{alg:FindHamiltonian}
\begin{algorithmic}[1]
\Require {A graph $G$ with $\abs{V(G)}=n$ and $\delta(G) \geq \lb$}
\Ensure {$C$ is a cycle of $G$}

\If {$G$ has a cut vertex} \Return NONE. \EndIf
\State $\bar{G} \gets$ the complement of $G$.
\State $\bar{H} \gets$ the biggest connected component of $\overline{G}$.
\If {$\bar{H}$ is a complete graph and $\abs{V(\bar{H})} \ge \frac{n}{2}$} \Return NONE. \EndIf
\State $P \gets$ a trivial path in $G$.

\Repeat
\While{$P$ is not maximal}
\State Append an edge to $P$.
\Comment{$P$ is a maximal path in $G$.}
\EndWhile

\State $C \gets $ \Call{MakeCycle}{$G,P$}
\If{$C \neq$ NONE and $|V(C)| \neq |V(G)|$}
\State Let $e$ be an edge with exactly one endpoint in $C$.
\State Let $e'$ be an edge of $C$ incident to $e$ \Comment{There are two such edges.}
\State $P \gets C + e - e'$
\EndIf
\Until{$|V(C)|=|V(H)|$ or $C=$ NONE}
\State \Return $C$. \Comment{$C$ is a Hamiltonian cycle of $G$.}
\Statex
\Function{MakeCycle}{$G,P$}
\Require{$P$ is a maximal path in $G$.}
\Ensure{return a cycle $C$ such that $V(C)=V(P)$ or ``NONE"}
\State $C \gets $ \Call{MakeTypeACycle}{$G,P$}
\If {$C \neq$ NONE} \Return $C$. \EndIf
\State $C \gets $ \Call{MakeTypeBCycle}{$G,P$}
\If {$C \neq$ NONE} \Return $C$. \EndIf
\State $C \gets $ \Call{MakeTypeCCycle}{$G,P$}
\If {$C \neq$ NONE} \Return $C$. \EndIf
\State \Return NONE.
\EndFunction
\end{algorithmic}
\end{algorithm}

\begin{algorithm}[H]
\footnotesize
\caption{Making a Type-A Cycle}\label{alg:MakeTypeACycle}
\begin{algorithmic}[1]
\Function{MakeTypeACycle}{$G,P$}
\Require{$P$ is a maximal  path in $G$.}
\Ensure{return a cycle $C$ such that $V(C)=V(P)$ or ``NONE"}
\State Let $P=x_0x_1\ldots x_k$.
\For{$i \in [0,k-3]$}
\For{$j \in [i+2,k-1]$}
\If{$x_0x_{i+1} \in E(G)$ and  $x_ix_{j+1} \in E(G)$ and $x_jx_k \in E(G)$} \label{lin:TypeACondition}
\State \Return $C=(x_0,x_1,\ldots,x_i,x_{j+1},x_{j+2},\ldots,x_k,x_j,x_{j-1}\ldots,x_{i+1},x_0)$.
\EndIf
\EndFor
\EndFor
\State \Return NONE.
\EndFunction
\end{algorithmic}
\end{algorithm}
\begin{algorithm}[H]
\footnotesize
\caption{Making a Type-B Cycle}\label{alg:MakeTypeBCycle}
\begin{algorithmic}[1]
\Function{MakeTypeBCycle}{$G,P$}
\Require{$P$ is a maximal path in $G$.}
\Ensure{return a cycle $C$ such that $V(C)=V(P)$ or ``NONE"}

\State Let $P=x_0x_1\ldots x_k$.
\For{$i \in [k-2]$}
\If{$x_0x_{i+1} \in E(G)$ and $x_{i-1}x_k \in E(G)$} \label{lin:TypeBCondition1}
\For{$j \in [k-1] \setminus \set{i}$}
\If{$x_ix_j \in E(G)$ and  $x_ix_{j+1} \in E(G)$} \label{lin:TypeBCondition2}
\State \Return $C=(x_0,x_1,\ldots,x_{i-1},x_k,x_{k-1},\ldots,x_{j+1},x_i,x_j,x_{j-1}\ldots,x_{i+1},x_0)$.
\EndIf
\EndFor
\EndIf
\EndFor
\State \Return NONE.
\EndFunction
\end{algorithmic}
\end{algorithm}
\begin{algorithm}[H]
\footnotesize
\caption{Making a Type-C Cycle}\label{alg:MakeTypeCCycle}
\begin{algorithmic}[1]
\Function{MakeTypeCCycle}{$G,P$}
\Require{$P$ is a maximal path in $G$.}
\Ensure{return a cycle $C$ such that $V(C)=V(P)$ or ``NONE"}

\State Let $P=x_0x_1\ldots x_k$.
\If{$x_0 x_{k-1} \in E(G)$}
\For{$i \in [0,k-2]$}
\If{$x_k x_i \in E(G)$ and $x_k x_{i+1} \in E(G)$ }
\State \Return $C=(x_0,x_1,\ldots,x_i,x_k,x_{i+1},\ldots,x_{k-1},x_0)$
\EndIf
\EndFor
\EndIf
\If{$x_k x_1 \in E(G)$}
\For{$i \in [k-1]$}
\If{$x_0 x_i \in E(G)$ and $x_0 x_{i+1} \in E(G)$ }
\State \Return $C=(x_1,x_2,\ldots,x_i,x_0,x_{i+1},\ldots, x_k, x_1)$
\EndIf
\EndFor
\EndIf
\State \Return NONE.
\EndFunction
\end{algorithmic}
\end{algorithm}
\comment{We can generalize type A, to paths of any length, not necessarily 3}
\comment{We can generalize type B, by replacing $x_i$ with a path of any length}
\comment{We can generalize type C, by replacing $x_0$ (resp. $x_k$) with a path of any length}

\section{Conclusion}\label{sec:conclusion}

In this work, we presented an extension of the classical Dirac theorem to the case where $\delta(G) \ge \lb$. We identified the only non-Hamiltonian graph families under this minimum degree condition. Our proof is short, simple, constructive, and self-contained. Then, we provided a polynomial-time algorithm that constructs a Hamiltonian cycle, if exists, of a graph $G$ with $\delta(G) \ge \lb$, or determines that the graph is non-Hamiltonian. The proof we present for the algorithm provides insight into the pattern of vertices on Hamiltonian cycles when $\delta(G) \ge \lb$. We believe that this insight will be useful in extending the results of this paper to graphs with lower minimum degrees, i.e., in identifying the exceptional non-Hamiltonian graph families when the minimum degree is smaller and constructing the Hamiltonian cycles, if exists. A natural question to ask in this direction is: What are the exceptional non-Hamiltonian graph families when $\delta(G) \le (n-1)/2$, $\delta(G)\le \lfloor (n-1)/2 \rfloor$ or $\delta(G) \le (n-2)/2$? How can we design an algorithm that not only determines whether a Hamiltonian cycle exists in such a case, but also constructs the Hamiltonian cycle whenever it exists. The investigation of these questions is subject of future work.

\end{document}